\documentclass{article}

\usepackage[square,sort&compress,numbers]{natbib}

\usepackage[australian]{babel}
\usepackage{fullpage}

\usepackage{hyperref}

\usepackage{amssymb,amsmath,amsfonts}
\usepackage{amsthm}
\usepackage[all]{xy}
\usepackage[capitalise]{cleveref}

\allowdisplaybreaks

\usepackage{xcolor}
\definecolor{MyBlue}{cmyk}{1,0.13,0,0.63}
\definecolor{MyGreen}{cmyk}{0.91,0,0.88,0.52}
\newcommand{\mylinkcolor}{MyBlue}
\newcommand{\mycitecolor}{MyGreen}
\newcommand{\myurlcolor}{webbrown}

\makeatletter
\def\@endtheorem{\endtrivlist}% NEW
\makeatother
\theoremstyle{plain}
\newtheorem{thm}{Theorem}[section]

\newtheorem{prop}[thm]{Proposition}

\theoremstyle{definition}
\newtheorem{defn}[thm]{Definition}
\newtheorem{remark}[thm]{Remark}

\renewcommand{\eqref}[1]{\labelcref{#1}}
\crefname{thm}{Theorem}{Theorems}
\crefname{lem}{Lemma}{Lemmas}
\crefname{prop}{Proposition}{Propositions}
\crefname{coro}{Corollary}{Corollaries}
\crefname{defn}{Definition}{Definitions}
\crefname{example}{Example}{Examples}
\crefname{remark}{Remark}{Remarks}

\hypersetup{%
  bookmarksnumbered=true,bookmarksopen=false,%
  plainpages=false,% necessary to prevent duplicate page identifiers
  linktocpage=true,%
  colorlinks=true,breaklinks=true,%
  linkcolor=\mylinkcolor,citecolor=\mycitecolor,urlcolor=\myurlcolor,%
  pdfpagelayout=OneColumn,%
  pageanchor=true,%
}

\usepackage{enumitem}
\setlist{topsep=4pt plus 2pt minus 2pt,partopsep=0pt,itemsep=2pt plus 2pt minus 2pt,parsep=0.5\parskip}
\setenumerate[1]{label=\arabic*)}

\RequirePackage{slashed}
\newcommand{\sD}{\slashed{D}}

\RequirePackage{amssymb,amsmath}

\newcommand{\R}{\mathbb{R}}
\newcommand{\C}{\mathbb{C}}
\newcommand{\qH}{\mathbb{H}}
\newcommand{\Z}{\mathbb{Z}}
\newcommand{\A}{\mathcal{A}}
\newcommand{\mH}{\mathcal{H}}
\newcommand{\mK}{\mathcal{K}}
\newcommand{\mJ}{\mathcal{J}}
\newcommand{\D}{\mathcal{D}}
\newcommand{\G}{\mathcal{G}}

\newcommand{\B}{\mathcal{B}}

\newcommand{\F}{\mathcal{F}}
\newcommand{\mF}{\mathcal{F}}

\newcommand{\E}{\mathcal{E}}
\newcommand{\mL}{\mathcal{L}}

\newcommand{\mU}{\mathcal{U}}

\newcommand{\su}{\mathfrak{su}}

\DeclareMathOperator{\Dom}{Dom}

\DeclareMathOperator{\Ker}{Ker}

\DeclareMathOperator{\cc}{c.c.}
\DeclareMathOperator{\Pert}{Pert}

\renewcommand{\Re}{\mathop{\textnormal{Re}}}
\renewcommand{\Im}{\mathop{\textnormal{Im}}}
\renewcommand{\bar}[1]{\overline{#1}}
\newcommand{\Cliff}{{\mathrm{Cl}}}

\newcommand{\dvol}{\textnormal{dvol}}
\newcommand{\op}{\textnormal{op}}

\newcommand{\Sub}[1]{_{\scriptscriptstyle#1}}
\newcommand{\til}[1]{\widetilde{#1}}
\newcommand{\tilpi}{\tilde\pi}

\newcommand{\hotimes}{\mathbin{\hat\otimes}}
\newcommand{\hot}{\hotimes}

\newcommand{\la}{\langle}
\newcommand{\ra}{\rangle}
\newcommand{\into}{\hookrightarrow}
\newcommand{\mvert}{\,|\,}

\newcommand{\bundlefont}[1]{{\mathtt{#1}}}
\newcommand{\bS}{\bundlefont{S}}
\newcommand{\bE}{\bundlefont{E}}

\newcommand{\mattwo}[4]{
  \left(\!\!\!\begin{array}{c@{~}c}#1&#2\\ #3&#4\\\end{array}\!\!\!\right)
}
\newcommand{\vectwo}[2]{
  \left(\!\!\!\begin{array}{c}#1\\ #2\\\end{array}\!\!\!\right)
}
\newcommand{\matthree}[3]{
  \left(\!\!\begin{array}{c@{~}c@{~}c}#1\\ #2\\ #3\\\end{array}\!\!\right)
}
\newcommand{\matfour}[4]{
  \left(\!\!\begin{array}{c@{~}c@{~}c@{~}c}#1\\ #2\\ #3\\ #4\\\end{array}\!\!\right)
}

\title{Krein spectral triples and the fermionic action}
\author{
Koen van den Dungen$^{1,2,3}$%
\footnote{%
Electronic mail: \texttt{kdungen@sissa.it}
}\\[4mm]
{\normalsize ${}^1$Mathematical Sciences Institute, Australian National University}\\
{\normalsize Canberra, ACT 0200, Australia}\\[2mm]
{\normalsize ${}^2$School of Mathematics and Applied Statistics, University of Wollongong}\\
{\normalsize Wollongong, NSW 2522, Australia}\\[2mm]
{\normalsize ${}^3$\emph{Current address:} SISSA (Scuola Internazionale Superiore di Studi Avanzati)}\\ 
{\normalsize Via Bonomea, 265, 34136 Trieste, Italy}
}

\date{}

\begin{document}

\maketitle

\begin{abstract}
\noindent
Motivated by the space of spinors on a Lorentzian manifold, we define Krein spectral triples, which generalise spectral triples from Hilbert spaces to Krein spaces. This Krein space approach allows for an improved formulation of the fermionic action for almost-commutative manifolds. We show by explicit calculation that this action functional recovers the correct Lagrangians for the cases of 
electrodynamics, 
the electro-weak theory, and the Standard Model. 
The description of these examples does not require a real structure, unless one includes Majorana masses, in which case the internal spaces also exhibit a Krein space structure. 

\vspace{\baselineskip}\noindent
\emph{Keywords}: Lorentzian manifolds, noncommutative geometry, gauge theories.

\noindent
\emph{Mathematics Subject Classification 2010}: 53C50, 58B34, 70S15. 
\end{abstract}

\section{Introduction}

The framework of Connes' noncommutative geometry \cite{Connes94}, and in particular the special case of so-called almost-commutative manifolds \cite{ISS04}, 
can be used to derive physical models describing both gravity and (classical) gauge theory, thus providing a first step towards a unified theory. For a suitably chosen almost-commutative manifold, one obtains the full Standard Model of high energy physics, including the Higgs mechanism and neutrino mixing \cite{Con96,CCM07}. 
This unified description of gauge theory and gravity relies on two action functionals which allow to \emph{derive} the Lagrangian of the theory: the \emph{spectral action} \cite{CC97} and the \emph{fermionic action} \cite{Con06}. 
The
spectral action yields the bosonic part of the Lagrangian, while the fermionic action yields (of course) the fermionic part (including the interactions between fermions and bosons). 

In this article we will focus on the fermionic action. The usual fermionic action, given by Connes \cite{Con06}, is given for a real even spectral triple $(\A,\mH,\D,J,\gamma)$ of $KO$-dimension $2$ as
$$
S_f := \frac12 \langle J\til\xi \mid \D\til\xi\rangle ,
$$
where $\til\xi$ is a Grassmann variable corresponding to an even vector $\xi=\gamma\xi\in\mH^+$. 
While the Lagrangian obtained from this action closely resembles (term by term) the physical Lagrangian, 
there are still two discrepancies. First, the fermionic action is given in Riemannian signature, while physical spacetime has Lorentzian signature. Second, the fermionic action is defined by using the real structure (the `charge conjugation operator'), but (except for possible Majorana mass terms) the charge conjugation operator should not be present in the physical Lagrangian. In addition, the presence of this real structure also means that the fermionic action is not automatically real-valued. 

In this article we will consider a Lorentzian version of noncommutative geometry. The study of Lorentzian noncommutative geometry is still far from complete, although some progress has been made. While the spinors on a Riemannian manifold give rise to a Hilbert space, the spinors on a pseudo-Riemannian (e.g.\ Lorentzian) manifold naturally give rise to a Krein space instead. Thus, to obtain a physical description of spinor fields on spacetime, it is more natural to work with Krein spaces instead of Hilbert spaces. A Krein space approach to noncommutative geometry was first taken in \cite{Str06} and studied further in \cite{Sui04,PS06}. Building upon these previous works, we propose a definition of \emph{Krein spectral triples}, which offers a natural extension of the notion of spectral triples from Hilbert spaces to Krein spaces. 

Subsequently, we will describe a \emph{Lorentzian} alternative to the fermionic action, which we call the \emph{Krein action}. As observed by Barrett \cite{Bar07}, the Krein action can be chosen to simply take the form of the usual Dirac action, and is given by
$$
S_\mK[\psi] := \la\psi|\D\psi\ra ,
$$
where $\la\cdot|\cdot\ra$ denotes the indefinite inner product on a Krein space $\mH$, and $\D$ is an unbounded Krein-self-adjoint operator on $\mH$. Thus, this action  by construction allows for Lorentzian signatures, and it does not involve the charge conjugation.
Unlike Barrett, however, 
we avoid the introduction of anti-commuting Grassmann variables for the fermions, which we find more natural, as an almost-commutative manifold only describes a classical gauge theory (not a quantum theory). 
Another difference is that, when Majorana masses are included, our finite space also exhibits a Krein space structure (see Section \ref{sec:Majorana}).

In \cite{NW96}, a continuous Wick rotation for spinor fields has been introduced, which could be used to turn the Krein action into a Euclidean action. Although a detailed study of such Wick rotations falls beyond the scope of this article, we would like to point out that such a Wick rotation of our Krein action cannot correspond to Connes' fermionic action, due to the presence of the charge conjugation operator in the latter. 

Lastly, let us point out that our description of gauge theories in terms of Krein spectral triples is not yet complete, as we have no Krein alternative for the bosonic action. It is still very much an open question how the spectral action should be adapted to the Lorentzian setting.

The layout of this article is as follows. First, we describe in Section \ref{sec:krein_ST} the definition of Krein spectral triples and the Krein action. 
We develop the general formalism for describing gauge theories from almost-commutative Lorentzian manifolds in Sections \ref{sec:krein_gauge} and \ref{sec:krein_acm}. Subsequently, we will calculate the Krein action explicitly for the cases of 
electrodynamics (Section \ref{sec:krein_ED}), 
the electro-weak theory (Section \ref{sec:krein_EW}), and the Standard Model (Section \ref{sec:krein_SM}). 
In these examples we will see that (unless one wants Majorana masses) it is not necessary to include the `anti-particles' in the finite Hilbert space. Majorana masses can be added by including these `anti-particles' (i.e.\ by doubling the finite Hilbert space), and by turning this doubled space into a Krein space (see Section \ref{sec:Majorana}).

\subsection*{Acknowledgements}
The author wishes to thank Adam Rennie and Walter van Suijlekom for useful comments and discussions. The author also thanks the referee for comments and suggestions for improvement. 
The author acknowledges support from both the 
Australian National University and the University of Wollongong.

\section{Krein spectral triples}
\label{sec:krein_ST}

A Krein space is a vector space $\mH$ with a non-degenerate inner product $\la\cdot|\cdot\ra$ which admits a fundamental decomposition $\mH = \mH^+ \oplus \mH^-$ (i.e., an orthogonal direct sum decomposition into a positive-definite subspace $\mH^+$ and a negative-definite subspace $\mH^-$) such that $\mH^+$ and $\mH^-$ are intrinsically complete (i.e., complete with respect to the norms $\Vert \psi\Vert_{\mH^\pm} := |\la\psi|\psi\ra|^{1/2}$).

A fundamental symmetry $\mJ$ is a self-adjoint unitary operator $\mJ\colon\mH\to\mH$ such that $(1+\mJ)\mH$ is positive-definite and $(1-\mJ)\mH$ is negative-definite. Given a fundamental decomposition $\mH = \mH^+ \oplus \mH^-$, we obtain a corresponding fundamental symmetry $\mJ = P^+ - P^-$, where $P^\pm$ denotes the projection onto $\mH^\pm$. 
Given a fundamental symmetry $\mJ$, we denote by $\mH_\mJ$ the corresponding Hilbert space for the positive-definite inner product $\la\cdot|\cdot\ra_\mJ := \la\mJ\cdot|\cdot\ra$. 

For an operator $T$, we will denote by $T^+$ the \emph{Krein-adjoint} (i.e., the adjoint operator with respect to the Krein inner product $\la\cdot|\cdot\ra$). By the \emph{adjoint} $T^*$ we will mean the usual adjoint in the Hilbert space $\mH_\mJ$ (i.e., with respect to the positive-definite inner product $\la\cdot|\cdot\ra_\mJ$). These adjoints are related via $T^+ = \mJ T^* \mJ$. For a detailed introduction to Krein spaces, we refer to \cite{Bognar74}. 

\begin{defn}
A Krein space $\mH$ with fundamental symmetry $\mJ$ is called \emph{$\Z_2$-graded} if $\mH_\mJ$ is $\Z_2$-graded and $\mJ$ is homogeneous. 
\end{defn}

The assumption that $\mH_\mJ$ is $\Z_2$-graded 
means we have a decomposition $\mH^0\oplus\mH^1$, and that this decomposition is respected by the positive-definite inner product $\la\cdot|\cdot\ra_\mJ$ (which means that $\la\psi_0|\psi_1\ra_\mJ = 0$ for all $\psi_0\in\mH^0$ and $\psi_1\in\mH^1$). The bounded operators $\B(\mH)$ then also decompose into a direct sum of even operators $\B^0(\mH)$ and odd operators $\B^1(\mH)$. The assumption that the fundamental symmetry $\mJ$ is homogeneous means that $\mJ$ is either even or odd. If $\mJ$ is odd, it implements a unitary isomorphism $\mH^0\simeq\mH^1$. Given the decomposition $\mH^0\oplus\mH^1$, we have a (self-adjoint, unitary) grading operator $\Gamma$ which acts as $(-1)^j$ on $\mH^j$ (for $j\in\Z_2$). If $\mJ$ is odd, we note that $\Gamma$ is Krein-\emph{anti}-self-adjoint (indeed, $\Gamma^+ = \mJ\Gamma\mJ = -\Gamma\mJ^2 = -\Gamma$). 

As in \cite[\S2.1]{vdDR16} we define the `combined graph inner product' $\la\cdot|\cdot\ra_{S,T}$ of two closed operators $S$ and $T$ as $\la\psi|\phi\ra_{S,T} := \la\psi|\phi\ra_\mJ + \la S\psi|S\phi\ra_\mJ + \la T\psi|T\phi\ra_\mJ$ (using the \emph{positive-definite} inner product $\la\cdot|\cdot\ra_\mJ$), for all $\psi,\phi\in\Dom S\cap\Dom T$. This inner product yields the corresponding `combined graph norm' $\|\cdot\|_{S,T}$. 
For a Krein-self-adjoint operator $\D$ we have $\mJ\D^* = \D\mJ$ and $\Dom\D^* = \Dom\D\mJ = \mJ\cdot\Dom\D$. 
One can then check that $\la\cdot|\cdot\ra_{\D,\D^*}$ is identical to $\la\cdot|\cdot\ra_{\D\mJ,\mJ\D}$ on $\Dom\D\cap\Dom\D^* = \Dom\D\cap\mJ\cdot\Dom\D$. 

The following definition aims to adapt the notion of spectral triple to Krein spaces. Similar approaches for such an adaptation have been given in \cite{PS06,Str06}. Our definition is similar to the notion of indefinite spectral triple from \cite{vdDR16}, except for the addition of the fundamental symmetry, and the replacement of the conditions on $\Re\D$ and $\Im\D$ by the \emph{Krein}-self-adjointness of $\D$. 
We will focus only on the even case; the odd case is defined similarly by removing the $\Z_2$-grading. 

\begin{defn}
\label{defn:Krein_triple}
An even \emph{Krein spectral triple} $(\A,\mH,\D,\mJ)$ consists of 
\begin{itemize}
\item a $\Z_2$-graded Krein space $\mH$; 
\item a trivially graded $*$-algebra $\A$ along with an even $*$-algebra representation $\pi\colon\A\to B^0(\mH)$;  
\item a fundamental symmetry $\mJ$ (satisfying $\mJ^*=\mJ$ and $\mJ^2=1$) which commutes with the algebra $\A$ and which is either even or odd; 
\item a densely defined, closed, odd operator $\D\colon\Dom\D\to\mH$ such that:
\begin{enumerate}
\item the linear subspace $\E:=\Dom\D\cap\mJ\cdot\Dom\D$ is dense in $\mH$;
\item the operator $\D$ is Krein-self-adjoint on $\E$ (or, equivalently, the operator $\mJ\D\colon\Dom\D\to\mH_\mJ$ is self-adjoint);
\item we have the inclusion $\pi(\A)\cdot\E\subset\E$, and the commutator $[\D,\pi(a)]$ is bounded on $\E$ for each $a\in\A$;
\item the map $\pi(a)\circ\iota\colon \E\into \mH\to \mH$ is compact for each $a\in \A$, where $\iota$ denotes the natural inclusion map $\E\into \mH$, and $\E:=\Dom\D\cap\mJ\cdot\Dom\D$ is considered as a Hilbert space with the inner product $\la\cdot|\cdot\ra_{\D\mJ,\mJ\D}$.
\end{enumerate}
\end{itemize}
We say an even Krein spectral triple $(\A,\mH,\D,\mJ)$ is of \emph{Lorentz-type} when $\mJ$ is \emph{odd}.
\end{defn}

The assumption that the fundamental symmetry commutes with the algebra need perhaps not be necessary in the general noncommutative case, but it will be satisfied by all examples we consider. This assumption ensures that we have $\pi(a)^+=\pi(a)^*=\pi(a^*)$ for all $a\in\A$. 

We have aimed to give a `complete' definition for Krein spectral triples, which can be used as a framework for the study of Lorentzian noncommutative geometry. However, this framework will undoubtedly undergo further transformations, and the above definition may not be definitive. We point out that, for the description of the fermionic action, only conditions 2.\ and 3.\ for the operator $\D$ are relevant. Nevertheless, we will show that the almost-commutative manifolds constructed in Section \ref{sec:krein_acm} satisfy all conditions, also for non-compact manifolds (i.e., non-unital algebras). 

As follows from Proposition \ref{prop:mfd_trip} below, the assumption that $\mJ$ is odd actually just captures the fact that the number of time dimensions is odd; it does not necessarily imply that there is only one time dimension (as was noted already in \cite[page 5]{PS06}).

A \emph{quadratic form} on a Krein space $\mH$ is a sesquilinear map $q\colon\Dom q\times\Dom q\to\C$ (conjugate-linear in the first variable and linear in the second variable), where the \emph{form domain} $\Dom q$ is a dense linear subspace of $\mH$. If $q(\psi_1,\psi_2) = \bar{q(\psi_2,\psi_1)}$ for all $\psi_1,\psi_2\in\mH$ we say that $q$ is \emph{symmetric}. If $\mH = \mH^0\oplus\mH^1$ is $\Z_2$-graded (and we think of $\C$ as being trivially graded), then we say that $q$ is \emph{$\Z_2$-graded} if $q(\psi_0,\psi_1)=0$ for any $\psi_0\in\mH^0\cap\Dom q$ and $\psi_1\in\mH^1\cap\Dom q$. 

\begin{prop}
Let $(\A,\mH,\D,\mJ)$ be an even Krein spectral triple. Then 
\begin{align*}
\F(\psi_1,\psi_2) := \left\langle\psi_1|\D\psi_2\right\rangle = \left\langle \mJ\psi_1|\D\psi_2\right\rangle_{\mJ} 
\end{align*}
defines a symmetric quadratic form $\mF$ with form domain $\Dom\mF = \Dom\D$. Moreover, if the Krein spectral triple is of Lorentz-type, then $\F$ is $\Z_2$-graded. 
\end{prop}
\begin{proof}
Sesquilinearity is immediate from the definition, and using the Krein-self-adjointness of $\D$ we also find symmetry: $\bar{\langle \psi_1|\D\psi_2\rangle} = \langle \D\psi_2|\psi_1\rangle = \langle \psi_2|\D\psi_1\rangle$.
If the triple is of Lorentz-type, then the grading operator $\Gamma$ is Krein-anti-self-adjoint. 
For $\psi_0\in\mH^0\cap\Dom\D$ and $\psi_1\in\mH^1\cap\Dom\D$ we then find that $\langle \psi_0|\D\psi_1\rangle = \langle \Gamma\psi_0|\D\psi_1\rangle = - \langle \psi_0|\Gamma \D\psi_1\rangle = \langle \psi_0|\D\Gamma\psi_1\rangle = - \langle \psi_0|\D\psi_1\rangle$. 
\end{proof}

\begin{defn}[Krein action]
\label{defn:Krein_action}
Let $(\A,\mH,\D,\mJ)$ be a Lorentz-type spectral triple. We define the \emph{Krein action} $S_\mK\colon\mH^0\to\C$ to be the functional
\begin{align*}
S_\mK[\psi] &:= \F(\psi,\psi) = \la\psi|\D\psi\ra . 
\end{align*}
\end{defn}
Since $\F$ is a symmetric quadratic form,  the Krein action $S_\mK[\psi]$ is automatically real-valued. 
Furthermore, we will show in 
Sections \ref{sec:krein_ED}, \ref{sec:krein_EW} and \ref{sec:krein_SM} 
that this Krein action recovers the correct (fermionic part of the) Lagrangians for 
electrodynamics, 
the electro-weak theory, and the Standard Model.

\section{Gauge theory}
\label{sec:krein_gauge}

In this section, we develop the abstract formalism for a description of gauge theories using Krein spectral triples. Later, we will apply this to the special case of almost-commutative manifolds. 

Let $\A$ be a trivially graded unital $*$-algebra. Denote by $\A^\op := \{ a^\op \mid a\in\A\}$ the \emph{opposite algebra} of $\A$, which equals $\A$ as a vector space but has the opposite product $a^\op b^\op = (ba)^\op$. 
Let $\mH$ be a $\Z_2$-graded Krein space with fundamental symmetry $\mJ$, and suppose we have two commuting even representations $\pi\colon \A\to\B^0(\mH)$ and $\pi^\op\colon\A^\op\to\B^0(\mH)$. For ease of notation, we will often simply write $a$ instead of $\pi(a)$ and $a^\op$ instead of $\pi^\op(a^\op)$. 
We obtain a representation of the algebraic tensor product $\A\odot\A^\op$ on $\mH$ by setting $\tilpi(a\otimes b^\op) := \pi(a)\pi^\op(b^\op)$. Now suppose that $(\A\odot\A^\op,\mH,\D,\mJ)$ is a Krein spectral triple. We say that this triple satisfies the \emph{order-one condition} if 
\begin{align}
\label{eq:order-one}
\big[\pi(a),[\D,\pi^\op(b^\op)]\big] = 0 
\end{align}
for all $a,b\in\A$. 
In the remainder of this section we consider an even Krein spectral triple $(\A\odot\A^\op,\mH,\D,\mJ)$. We will assume that this triple is \emph{unital}, which means that $\A$ is unital and that $\tilpi$ is unital. The examples we describe in Sections \ref{sec:krein_ED}-\ref{sec:krein_SM} all satisfy the order-one condition. However, the order-one condition can fail in other examples, such as the Pati-Salam model \cite{CCvS13b}, and therefore we will develop the abstract formalism below without assuming that the triple satisfies the order-one condition. 

\subsection{Inner perturbations}

We will now introduce fluctuations of the operator $\D$, which will give rise to gauge fields as well as scalar fields (the latter being interpreted as the Higgs field in the case of electroweak theory and the Standard Model, see Sections \ref{sec:krein_EW} and \ref{sec:krein_SM}). 
We adapt the approach described in \cite{CCvS13} to our Krein spectral triples. 

Let $\A$ be a trivially graded unital $*$-algebra, and consider the algebraic tensor product $\A\odot\A^\op$. For an element $A = \sum_j a_j\otimes b_j^\op \in \A\odot\A^\op$, we have an anti-linear involution $A\to\bar A$ given by $\bar{\sum_j a_j\otimes b_j^\op} := \sum b_j^*\otimes a_j^{*\op}$, with the properties $\bar{(\lambda A)} = \bar\lambda\, \bar A$, $\bar{\bar A} = A$, and $\bar{(AA')} = \bar{A}\,\bar{A'}$, for all $\lambda\in\C$ and $A,A'\in\A\odot\A^\op$. 
We say that $A = \sum_j a_j\otimes b_j^\op\in\A\odot\A^\op$ is \emph{real}\footnote{In \cite{CCvS13}, $\bar A$ (denoted $A^*$ therein) is called the adjoint of $A$, and real elements are called self-adjoint. However, since $A\mapsto\bar A$ is an involution instead of an anti-involution, we prefer to think of this map as a `real structure', which turns the complex algebra $\A\odot\A^\op$ into a ``real'' algebra (following terminology of Kasparov \cite{Kas80b}).} if $\bar{A} = A$, and we say that $A$ is \emph{normalised} if $\sum a_j b_j = 1\in\A$.
The 
properties of reality and normalisation are preserved by multiplication, so we can define the following. 

\begin{defn}[{\cite[\S III]{CCvS13}}]
The \emph{perturbation semi-group} $\Pert(\A)$ is the set of real normalised elements in $\A\odot\A^\op$, with the multiplication inherited from the algebra $\A\odot\A^\op$.
\end{defn}

Let $(\B,\mH,\D,\mJ)$ be a Krein spectral triple. 
We consider the \emph{generalised one-forms} given by
$$
\Omega_\D^1(\B) := \Bigl\{ \sum_j a_j[\D,b_j] \Bigm\vert a_j,b_j\in\B \Bigr\} ,
$$
where the sums must converge in norm. 
For the algebra $\B = \A\odot\A^\op$, we then consider a map $\eta_\D\colon\A\odot\A^\op\to\Omega^1_\D(\A\odot\A^\op)\subset \B(\mH)$ defined by
$$
\eta_\D\Big(\sum_j a_j\otimes b_j^\op\Big) := \sum_{j,k} \tilpi(a_j\otimes (a_k^*)^\op) \big[\D , \tilpi(b_j\otimes (b_k^*)^\op)\big] .
$$
As in \cite[Lemma 4.(ii)]{CCvS13}, this map $\eta_\D$ is involutive (i.e.\ $\eta_\D(\bar A) = \eta_\D(A)^+$).
In particular, $\eta_\D$ maps elements of $\Pert(\A)$ to Krein-self-adjoint elements in $\Omega^1_\D(\A\odot\A^\op)$.
If we have a Krein spectral triple $(\A\odot\A^\op,\mH,\D,\mJ)$ which satisfies the order-one condition \eqref{eq:order-one}, then the expression for $\eta_\D$ simplifies to
\begin{align*}
\eta_\D\Big(\sum_j a_j\otimes b_j^\op\Big) &= \sum_{j} a_j [\D,b_j] + \sum_{j} a_j^{*\op} [\D,b_j^{*\op}]  .
\end{align*}

\begin{defn}
By the \emph{fluctuation} of $\D$ by $A\in\Pert(\A)$ we mean the map
$$
\D \mapsto \D_A := \D + \eta_\D(A) ,
$$
and we refer to $\D_A$ as the \emph{fluctuated Dirac operator}.
\end{defn}

We point out that the map $\eta_\D$ is \emph{not} multiplicative and therefore it does not yield a representation of the semi-group $\Pert(\A)$ on $\mH$. Instead, we obtain an action of $\Pert(\A)$ on the space of fluctuated Dirac operators.
\begin{prop}[{\cite[Proposition 5.(ii)]{CCvS13}}]
A fluctuation of a fluctuated Dirac operator is again a fluctuated Dirac operator. To be precise: $(\D_A)_{A'} = \D_{A'A}$ for all perturbations $A,A'\in\Pert(\A)$. 
\end{prop}

\subsection{Gauge action}

The unitary group $\mU(\A)$ maps to the perturbation semi-group $\Pert(\A)$ via the semi-group homomorphism $\Delta\colon\mU(\A)\to\Pert(\A)$ given by $u \mapsto u\otimes (u^*)^\op$. 
This then yields an obvious action of $u\in\mU(\A)$ on $\Pert(\A)$ given by multiplication with $\Delta(u)$. To be precise, for $A = \sum_j a_j\otimes b_j^\op \in \Pert(\A)$, the action of $u\in\mU(\A)$ is given by
$$
\Delta(u) A = \sum_j ua_j\otimes (u^*)^\op b_j^\op = \sum_j ua_j\otimes (b_ju^*)^\op . 
$$
Since $\Pert(\A)\subset\A\odot\A^\op$, we can compose $\Delta$ with the $*$-algebra representation $\tilpi$ to obtain a group representation 
$$
\rho := \tilpi\circ\Delta \colon \mU(\A) \to \B(\mH) .
$$

\begin{defn}
We define the \emph{gauge group} as 
$$
\G(\A) := \big\{ \rho(u) \mid u\in\mU(\A) \big\} \simeq \mU(\A) / \Ker\rho .
$$
\end{defn}

We also consider an action $\gamma$ of the unitary group $\mU(\A)$ on $\Omega^1_\D(\A\odot\A^\op)$. 
For $T\in\Omega^1_\D(\A\odot\A^\op)$ and $u\in\mU(\A)$, this action is given by 
$$
\gamma_u(T) := \rho(u) T \rho(u^*) + \eta_\D\circ\Delta(u) = \rho(u) T \rho(u^*) + \rho(u) [\D,\rho(u^*)] .
$$ 
We point out that this is the usual transformation of gauge potentials under the gauge group $\G(\A)$. 
The map $\eta_\D$ is covariant with respect to the actions by $\mU(\A)$ (cf.\ \cite[Lemma 4.(iii)]{CCvS13}). To be precise, $\gamma_u\circ\eta_\D(A) = \eta_\D(\Delta(u)A)$ for all $u\in\mU(\A)$.

As mentioned above, we have an action of $\G(\A)$ on $\mH$. We can also define an action of $\G(\A)$ on the space of fluctuated Dirac operators. For $\rho(u)\in\G(\A)$, this action is given by $\D_A\mapsto\D_{\Delta(u)A}$. 
If $u\in\Ker\rho$, then $\eta_\D(\Delta(u)A) = \eta_\D(A)$ and hence $\D_{\Delta(u)A} = \D_A$. Therefore this action of $\rho(u)$ on $\D_A$ is independent of the choice of representative $u\in\mU(\A)$ for $\rho(u)\in\G(\A)$. 

\begin{prop}
The Krein action $S_\mK[\psi,A] := \la\psi|\D_A\psi\ra$ of the fluctuated Dirac operator $\D_A$ is invariant under the action of the gauge group given by $\psi\mapsto\rho(u)\psi$ and $A \mapsto \Delta(u)A$. 
\end{prop}
\begin{proof}
Since $\eta_\D$ is covariant under the action of the gauge group, we find
\begin{align*}
\D_{\Delta(u)A} &= \D + \eta_\D(\Delta(u)A) 
= \D + \gamma_u\circ\eta_\D(A) 
= \D + \rho(u) \eta_\D(A) \rho(u^*) + \rho(u) [\D,\rho(u^*)] \\
&= \rho(u) \big( \D + \eta_\D(A) \big) \rho(u^*)
= \rho(u) \D_A \rho(u^*) . 
\end{align*}
Since the unitary $\rho(u)$ commutes with $\mJ$ (so $\rho(u)$ is also unitary for the Krein inner product), we find
\begin{align*}
S_\mK[\rho(u)\psi,\Delta(u)A] &= \big\la \rho(u)\psi \mvert \D_{\Delta(u)A}\rho(u)\psi \big\ra 
= \big\la \rho(u)\psi \mvert \rho(u)\D_A\psi \big\ra
= \la\psi \mvert \D_A\psi\ra = S_\mK[\psi,A] . \qedhere 
\end{align*}
\end{proof}

\section{Almost-commutative manifolds}
\label{sec:krein_acm}

A \emph{finite space} $F := (\A_F,\mH_F,\D_F,\mJ_F)$ with $\Z_2$-grading $\Gamma_F$ is an even Krein spectral triple for which $\mH_F$ is finite-dimensional. An almost-commutative manifold is then constructed as the product of a finite space with a pseudo-Riemannian spin manifold $M$.
Let us first have a closer look at the Krein spectral triple corresponding to such a manifold. 

Let $(M,g)$ be an $n$-dimensional time- and space-oriented pseudo-Rieman\-nian spin manifold of signature $(t,s)$, where $t$ is the number of time dimensions (for which $g$ is negative-definite) and $s$ is the number of spatial dimensions (for which $g$ is positive-definite). We consider an orthogonal decomposition of the tangent bundle $TM = \bE_t\oplus \bE_s$, which always exists but is far from unique. We will consider elements of $\bE_t$ to be `purely timelike' and elements of $\bE_s$ to be `purely spacelike'. Given our choice of decomposition $TM = \bE_t\oplus \bE_s$, we have a \emph{timelike projection} $T\colon \bE_t\oplus \bE_s \to \bE_t$ and a \emph{spacelike reflection} $r := 1-2T$ which acts as $(-1)\oplus1$ on $\bE_t\oplus \bE_s$. 
Using the spacelike reflection $r$, we can define a `Wick rotated' metric $g_r$ on $M$ by setting
$$
g_r(v,w) := g(rv,w)
$$
for all $v,w\in TM$. One readily checks that $g_r$ is positive-definite, and hence $(M,g_r)$ is a Riemannian manifold.

Let $\Cliff(TM,g)$ denote the real Clifford algebra with respect to $g$, and denote the Clifford representation $TM \into \Cliff(TM,g)$ by $\gamma$. Our conventions are such that $\gamma(v)\gamma(w) + \gamma(w)\gamma(v) = - 2 g(v,w)$. 
We shall denote by $h$ the map $T^*M\rightarrow TM$ which maps $\alpha\in T^*M$ to its dual in $TM$ with respect to the metric $g$. That is:
\begin{align*}
h(\alpha) &= v   \Longleftrightarrow   \alpha(w) = g(v,w) \quad\text{for all } w\in TM .
\end{align*}
We assume that $M$ is equipped with a spin structure. 
We consider the corresponding spinor bundle $\bS\rightarrow M$ 
and its space of compactly supported, smooth sections 
$\Gamma_c^\infty(\bS)$. We denote by $c$ the pseudo-Riemannian 
Clifford multiplication $\Gamma_c^\infty(T^*M\otimes \bS)\rightarrow \Gamma_c^\infty(\bS)$ 
given by 
$
c(\alpha\otimes\psi) := \gamma\big(h(\alpha)\big)\psi .
$
Let $\nabla$ be the Levi-Civita connection for the pseudo-Riemannian 
metric $g$, and let $\nabla^\bS$ be its lift to the spinor bundle. 
The Dirac operator on $\Gamma_c^\infty(\bS)$ is defined as the composition
\begin{align*}
\sD &\colon \Gamma_c^\infty(\bS) \xrightarrow{\nabla^\bS} \Gamma_c^\infty(T^*M\otimes \bS) \xrightarrow{c} \Gamma_c^\infty(\bS) . 
\end{align*}
Locally, we can choose a (pseudo-)orthonormal frame $\{e_j\}_{j=1}^n$ 
corresponding to our choice of decomposition $TM = \bE_t\oplus \bE_s$, 
such that $e_j\in \bE_t$ for $j\leq t$ and $e_j\in \bE_s$ for $j>t$. 
In terms of this frame, the metric can be written as
\begin{align*}
g(e_i,e_j) &= \delta_{ij}\kappa(j) , & \kappa(j) &= 
\begin{cases}-1 & j=1,\ldots,t;\\ 1 & j=t+1,\ldots,n. \end{cases}
\end{align*}
Let $\{\theta^i\}_{i=1}^n$ be the basis of $T^*M$ dual to $\{e_j\}_{j=1}^n$, 
so that $\theta^i(e_j) = \delta^i_j$. We then see that $h(\theta^j) = \kappa(j)e_j$. 
In terms of the local frame $\{e_j\}$, we can then write the Dirac operator as
$$
\sD := c \circ \nabla^\bS = \sum_{j=1}^n \kappa(j) \gamma(e_j) \nabla^\bS_{e_j} .
$$
Given the decomposition $TM = \bE_t\oplus \bE_s$, there exists a 
positive-definite hermitian structure \cite[\S3.3.1]{Baum81}
$$
(\cdot|\cdot)_{\mJ_M}\colon \Gamma_c^\infty(\bS) \times \Gamma_c^\infty(\bS) \to C_c^\infty(M) ,
$$
which gives rise to the inner product
$
\la\psi_1|\psi_2\ra_{\mJ_M} := \int_M (\psi_1|\psi_2)_{\mJ_M} \dvol_g ,
$
for all $\psi_1,\psi_2\in\Gamma_c^\infty(\bS)$, where $\dvol_g$ 
denotes the canonical volume form of $(M,g)$. The completion of 
$\Gamma_c^\infty(\bS)$ with respect to this inner product is denoted 
$L^2(\bS)$. 
We can define an operator $\mJ_M$ on $L^2(\bS)$ by setting
$$
\mJ_M := i^{t(t-1)/2} \gamma(e_1) \cdots \gamma(e_t) , 
$$
where $\{e_j\}$ is a local orthonormal frame corresponding to the 
decomposition $TM = \bE_t\oplus \bE_s$. This operator is self-adjoint 
and unitary, and is related to the spacelike reflection $r$ via
$
\mJ_M\gamma(v)\mJ_M = (-1)^t \gamma(rv) .
$
The space $L^2(\bS)$ then becomes a Krein space with the indefinite inner product $\la\cdot|\cdot\ra := \la\mJ_M\cdot|\cdot\ra_{\mJ_M}$ and with fundamental symmetry $\mJ_M$. This indefinite inner product $\la\cdot|\cdot\ra$ is independent of the choice of decomposition $TM = \bE_t\oplus \bE_s$. 

\begin{prop}
\label{prop:mfd_trip}
Let $(M,g)$ be an $n$-dimensional time- and space-oriented pseudo-Riemannian spin manifold of signature $(t,s)$. 
Let $r$ be a spacelike reflection, such that the associated Riemannian metric $g_r$ is complete. 
Then we obtain an even Krein spectral triple 
$$
\left( C_c^\infty(M) , L^2(\bS) , i^t\sD , \mJ_M \right) ,
$$
with grading operator $\Gamma_M$. If $t$ is odd, the triple is a Lorentz-type spectral triple. 
\end{prop}
\begin{proof}
We need to check that $i^t\sD$ satisfies the conditions of Definition \ref{defn:Krein_triple}. 
For condition 1.\ we note that the linear subspace $\E := \Dom \sD \cap \mJ_M\cdot\Dom\sD$ contains $\Gamma_c^\infty(\bS)$ and is therefore dense in $\mH$. 
Condition 2.\ follows from 
\cite[Satz 3.19]{Baum81}. 
Condition 3.\ follows because $\sD$ is a first-order differential operator and the algebra is smooth. 

To show the local compactness of the inclusion map $\E\into L^2(\bS)$, we consider the `Wick rotations' $\sD_\pm := \frac12(\sD+\sD^*) \mp \frac i2(\sD-\sD^*)$. We know from (the proof of) \cite[Proposition 4.2]{vdDR16} that $\sD_\pm$ are elliptic, and hence they have locally compact resolvents (see e.g.\ \cite[Proposition 10.5.2]{Higson-Roe00}).
The equality of the combined graph norms $\|\cdot\|_{\D,\D^*} = \|\cdot\|_{\D_+,\D_-}$ (see \cite[Lemma 2.3]{vdDR16}) shows that the inclusion $\E \into \Dom\sD_\pm$ is bounded, and therefore the composition $\E \into \Dom\sD_\pm \into L^2(\bS)$ is locally compact. 

Since $M$ is even-dimensional, we have 
a grading operator $\Gamma_M$ which satisfies the relation $\Gamma_M \mJ_M = (-1)^t \mJ_M \Gamma_M$ 
(see e.g.\ \cite[Ch.\ 1]{LM89}).  
This implies that we have a Lorentz-type spectral triple if and only if $t$ is odd. 
\end{proof}

Let $\mH$ be a $\Z_2$-graded Krein space. For a homogeneous element $\psi\in\mH^j$ ($j=0,1$), we write $|\psi|=j$ for the degree of $\psi$. Similarly, for a homogeneous operator $T$ on $\mH$, we have $|T|=0$ if $T$ is even, and $|T|=1$ if $T$ is odd. 

Given two $\Z_2$-graded Krein spaces $\mH_1$ and $\mH_2$, we define their graded tensor product $\mH_1\hot\mH_2$ as the vector space $\mH_1\otimes\mH_2$ with the inner product 
$$
\la \psi_1\hot\psi_2 | \phi_1\hot\phi_2 \ra := \la\psi_1|\phi_1\ra \la\psi_2|\phi_2\ra ,
$$
for $\psi_1,\phi_1\in\mH_1$ and $\psi_2,\phi_2\in\mH_2$, and with the grading given by $(\mH_1\hot\mH_2)^0 := (\mH_1^0\hot\mH_2^0) \oplus (\mH_1^1\hot\mH_2^1)$ and $(\mH_1\hot\mH_2)^1 := (\mH_1^0\hot\mH_2^1) \oplus (\mH_1^1\hot\mH_2^0)$. Given homogeneous operators $T_1$ on $\mH_1$ and $T_2$ on $\mH_2$, their tensor product acts on $\mH_1\hot\mH_2$ as
$$
(T_1\hot T_2) (\psi_1\hot\psi_2) := (-1)^{|T_2||\psi_1|} T_1\psi_1 \hot T_2\psi_2 ,
$$
for homogeneous elements $\psi_1\in\mH_1$ and $\psi_2\in\mH_2$. 
The adjoint of $T_1\hot T_2$ is given by
$$
(T_1\hot T_2)^* = (-1)^{|T_1||T_2|} T_1^*\hot T_2^* .
$$

\begin{defn}
Let $(M,g)$ be an even-dimensional pseudo-Riemannian spin manifold as in Proposition \ref{prop:mfd_trip}. 
An \emph{almost-commutative pseudo-Riemannian manifold} $F\times M$ is the product of a finite space $F$ with 
the manifold $M$, 
given by
\begin{multline*}
F\times M := (\A,\mH,\D,\mJ) := 
\left( C_c^\infty(M,\A_F) , \mH_F\hot L^2(\bS) , i^{|\mJ_F|}\hot i^t\sD + i^{|\mJ_M|}\D_F\hot1 , i^{|\mJ_F|\cdot|\mJ_M|}\mJ_F\hot\mJ_M \right) ,
\end{multline*}
equipped with the grading operator $\Gamma := \Gamma_F\hot\Gamma_M$. 
\end{defn}

As in \cite{BvdD14}, 
we construct almost-commutative manifolds as $F\times M$ instead of $M\times F$ (though the latter is more common in the literature). The reason is that the order $F \times M$ is more natural for the generalisation to the globally non-trivial case and its description as a Kasparov product (see \cite[\S III.C]{BvdD14}). 
Here we have also written the product in terms of the \emph{graded} tensor product $\hot$, thus avoiding explicit use of the grading operators.

\begin{prop}
\label{prop:acm_trip}
An almost-commutative pseudo-Riemannian manifold is an even Krein spectral triple. 
\end{prop}
\begin{proof}
The factor $i^{|\mJ_F|\cdot|\mJ_M|}$ in front of $\mJ_F\hot\mJ_M$ ensures that $\mJ$ is again self-adjoint and unitary. 
Furthermore, $\mJ$ commutes with the algebra $\A$ because $\mJ_F$ and $\mJ_M$ both commute with the algebra. Since $\mJ_F$ and $\mJ_M$ are homogeneous, $\mJ$ is also homogeneous with degree $|\mJ| = |\mJ_F|+|\mJ_M|$. 
The factors $i^{|\mJ_F|}$ and $i^{|\mJ_M|}$ before $1\hot i^t\sD$ and $\D_F\otimes1$ (respectively) ensure that $\D$ is again Krein-symmetric. The Krein-self-adjointness of $\D$ then follows from Krein-self-adjointness of $i^t\sD$ and boundedness of $\D_F$. 
The other properties of $\D$ follow straightforwardly from the properties of $\sD$ and $\D_F$. 
\end{proof}

Our typical example will be an almost-commutative manifold constructed from an even-dimensional Lorentzian manifold, which is of course of Lorentz-type (for which $\mJ_M$ is odd). In order to be able to apply the Krein action, we need this almost-commutative manifold to be of Lorentz-type as well, which means that the finite space should not be of Lorentz-type. Hence we impose the restriction that $\mJ_F$ is even. The almost-commutative manifold is then of the form
\begin{align}
\label{eq:ac-Lor}
F \times M := \left( C_c^\infty(M,\A_F) , \mH_F\hot L^2(\bS) , 1\hot i^t\sD + i\D_F\hot1 , \mJ_F\hot\mJ_M \right) .
\end{align}

\section{Electrodynamics}
\label{sec:krein_ED}

As a first example, we will calculate the Krein action for electrodynamics. The model of electrodynamics was first studied in the context of noncommutative geometry in \cite{vdDvS13}. Here, we take a slightly different approach, since we have no need for a real structure, and we can therefore reduce the dimension of the Hilbert space by a factor $2$. 

We consider the algebra $\A_F = \C\oplus\C$, and the even finite space
$$
F\Sub{ED} := \left( \A_F \odot \A_F^\op , \mH_F = \C^2 , \D_F = \mattwo{0}{-im}{im}{0} , \mJ_F = 1 \right) .
$$
We denote the standard basis of $\mH_F$ as $\{e_R,e_L\}$, where $e_R$ is odd and $e_L$ is even. 
Since $\A_F$ is commutative, we have $\A_F^\op\simeq \A_F = \C\oplus\C$. 
We consider the representations $\pi,\pi^\op\colon\C\oplus\C\to\B^0(\mH_F)$ given by
\begin{align*}
\pi(\lambda,\mu) &:= \lambda , & \pi^\op(\lambda,\mu) &:= \mu ,
\end{align*}
for $(\lambda,\mu)\in\C\oplus\C$, which gives the representation $\tilpi((\lambda,\mu)\otimes(\lambda',\mu')) = \lambda\mu'$ of $\A_F\odot \A_F^\op$ on $\mH_F$. We also note that these representations obviously satisfy the order-one condition \eqref{eq:order-one}. 
Since we have set $\mJ_F=1$, this finite space is in fact an ordinary finite spectral triple, and hence also a Krein spectral triple which is not of Lorentz-type. 

\begin{prop}
The gauge group of the finite space $F\Sub{ED}$ equals $\G(F\Sub{ED}) = U(1)$. 
\end{prop}
\begin{proof}
We have $\mU(\A_F) = U(1)\times U(1)$, and the kernel of the representation $\rho\colon\mU(\A_F)\to\B(\mH_F)$ equals $\Ker\rho = \{ (\lambda,\lambda)\in\mU(\A_F) \mid \lambda\in U(1) \} \simeq U(1)$, which yields for the quotient $\G(F\Sub{ED}) = \mU(\A_F) / \Ker\rho \simeq U(1)$. 
\end{proof}

Let $(M,g)$ be an even-dimensional pseudo-Riemannian spin manifold as in Proposition \ref{prop:mfd_trip}, for which $t$ is odd, and consider the corresponding almost-com\-mut\-ative manifold as in Eq.\ \ref{eq:ac-Lor}:
$$
F\Sub{ED} \times M := \left( C_c^\infty(M,\A_F\odot \A_F^\op) , \mH_F\hot L^2(\bS) , 1\hot i^t\sD + i\D_F\hot1 , 1\hot\mJ_M \right) .
$$

In the following, we will use local coordinates $x^\mu$ ($\mu=1,\ldots,n=t+s$) to write the Dirac operator as $\sD = \gamma^\mu \nabla^\bS_\mu$ for the `gamma matrices' $\gamma^\mu := \gamma(dx^\mu)$ and the spin connection $\nabla^\bS_\mu = \nabla^\bS_{\partial_\mu}$. 

We consider a perturbation $A\in\Pert(C_c^\infty(M,\A_F))$. Since the algebra $\A:= C_c^\infty(M,\A_F)$ is commutative, $\A^\op\simeq\A$ and we can write $A = \sum a_j\otimes b_j$ for $a_j=(\lambda_j,\mu_j)$ and $b_j=(\lambda_j',\mu_j')$ in $\A$. 
Since the order-one condition is satisfied, the expression for $\eta_\D$ simplifies and takes the form
$$
\eta_\D(A) = \sum_j \lambda_j [i^t\sD,\lambda_j'] + \sum_j \mu_j [i^t\sD,\mu_j'] =: A_\mu \hot i^t \gamma^\mu ,
$$
where we have used that $\D_F$ commutes with the algebra elements, and we have defined $A_\mu := \sum_j (\lambda_j\partial_\mu\lambda_j' + \mu_j\partial_\mu\mu_j') \in C_c^\infty(M)$. Since $A$ is real and $\eta_\D$ is involutive (cf.\ \cite[Lemma 4.(ii)]{CCvS13}), we know 
that $\eta_\D(A)$ is Krein-self-adjoint. Since $i^t\gamma^\mu$ is Krein-anti-symmetric, $A_\mu$ must also be Krein-anti-symmetric, and hence $A_\mu \in C_c^\infty(M,i\R)$. 
We consider the corresponding fluctuated Dirac operator given by
\begin{align}
\label{eq:fluc_Dirac_ED}
\D_A := 1 \hot i^t\sD + i\D_F \hot 1 + A_\mu \hot i^t\gamma^\mu .
\end{align}

The almost-commutative manifold $F\Sub{ED}\times M$ is of Lorentz-type, and hence we can apply Definition \ref{defn:Krein_action} for the Krein action. An arbitrary vector $\xi\in\mH^0 = \mH_F^0\hotimes L^2(\bS)^0 \oplus \mH_F^1\hotimes L^2(\bS)^1$ can be written uniquely as 
\begin{align}
\label{eq:right_fermion_ED}
\xi = e_R\hot\psi_R + e_L\hot\psi_L ,
\end{align}
for Weyl spinors $\psi_L\in L^2(\bS)^0$ and $\psi_R\in L^2(\bS)^1$. Note that the vector $\xi\in\mH^0$ is therefore completely determined by one Dirac spinor $\psi := \psi_L + \psi_R$. 

\begin{prop}
\label{prop:alt_fermion_act_ED}
The Krein action for $F\Sub{ED}\times M$ is given by
\begin{align*}
S\Sub{ED}[\psi,A] &= \big\langle \psi \mvert (i^t(\sD + \gamma^\mu A_\mu) - m)\psi\big\rangle .
\end{align*}
\end{prop}
\begin{proof}
We need to calculate the inner product $\la\mJ\xi|\D_A\xi\ra_\mJ$, where $\xi$ is given as in Eq.\ \ref{eq:right_fermion_ED}. First, for $\mJ = 1\hot\mJ_M$ we calculate
$$
\mJ\xi = - e_R\hot\mJ_M\psi_R + e_L\hot\mJ_M\psi_L .
$$
For the fluctuated Dirac operator $\D_A$ of Eq.\ \ref{eq:fluc_Dirac_ED} we find
\begin{align*}
\D_A\xi &= - e_R\hot i^t\sD\psi_R + e_L\hot i^t\sD\psi_L - m e_L\hot\psi_R + m e_R\hot\psi_L - A_\mu e_R\hot i^t\gamma^\mu\psi_R + A_\mu e_L\hot i^t\gamma^\mu\psi_L \\
&= - e_R \hot \big( i^t\sD\psi_R - m \psi_L + i^t\gamma^\mu A_\mu \psi_R \big) + e_L \hot \big( i^t\sD\psi_L - m \psi_R + i^t\gamma^\mu A_\mu \psi_L \big) .
\end{align*}
Taking the inner product of $\mJ\xi$ with $\sD_A\xi$, and using the orthogonality of $L^2(\bS)^0$ and $L^2(\bS)^1$, we obtain
\begin{align*}
\la\mJ\xi \mvert \D_A\xi\ra_{\mJ} 
&= \big\la\mJ_M\psi \,\big|\, i^t\sD\psi - m \psi + i^t\gamma^\mu A_\mu \psi \big\ra_{\mJ_M} . \qedhere
\end{align*}
\end{proof}

\begin{remark}
Let us consider the above result for the usual case of a $4$-dimen\-sional time- and space-oriented Lorentzian manifold $M$ of signature $(1,3)$. A choice of decomposition $TM = \bE_t\oplus\bE_s$ determines a unit timelike vector field $e_0\in\Gamma(\bE_t)$, because $\bE_t$ is oriented and has rank $1$. The fundamental symmetry then equals $\mJ_M = \gamma(e_0)$, and (using standard physics notation) we will write the (indefinite) inner product as $\la\psi|\phi\ra = \int_M \bar\psi\phi \, \dvol_g$, where $\bar\psi = \psi^\dagger\gamma(e_0)$ is the Dirac adjoint of $\psi$. We can then rewrite the Krein action 
as
$
S\Sub{ED}[\psi,A] = \int_M \mL\Sub{ED}[\psi,A] \,\dvol_g ,
$
where the Lagrangian 
is given by
$$
\mL\Sub{ED}[\psi,A] := \bar\psi \big(i\gamma^\mu(\nabla_\mu+A_\mu) - m\big)\psi .
$$
This is indeed \emph{precisely} the usual (fermionic part of the) Lagrangian for electrodynamics (compare, for instance, \cite[\S4.1]{Peskin-Schroeder95}).
\end{remark}

\section{The electro-weak theory}
\label{sec:krein_EW}

In this section we will describe the electro-weak interactions between leptons (i.e., neutrinos and electrons). The description given here is largely an adaptation of \cite[\S5]{vdDvS12}. 

Consider the finite-dimensional Hilbert space $\mH_F := \mH_R\oplus\mH_L$, where $\mH_R = \mH_L = \C^2$. 
This Hilbert space is $\Z_2$-graded with even part $\mH_F^0 = \mH_L$ and odd part $\mH_F^1 = \mH_R$. We denote the basis of $\mH_F$ by $\{\nu_R,e_R,\nu_L,e_L\}$, where the elements $\nu_R,\nu_L$ describe the (right- and left-handed) neutrinos, and $e_R,e_L$ describe the electrons. 

We consider the (real) algebra $\A_F=\C\oplus\qH$, along with two even (real-linear) representations $\pi\colon \A_F \to \B(\mH_R)\oplus \B(\mH_L)$ and $\pi^\op\colon \A_F^\op \to \B(\mH_R)\oplus \B(\mH_L)$ given by
\begin{align*}
\pi(\lambda,q) &:= q_\lambda \oplus q := \mattwo{\lambda}{0}{0}{\bar\lambda} \oplus \mattwo{\alpha}{\beta}{-\bar\beta}{\bar\alpha} , & 
\pi^\op\big((\lambda,q)^\op\big) &:= \lambda\oplus\lambda ,
\end{align*}
for $\lambda\in\C$ and $q = \alpha + \beta j \in\qH$. 
The representation $\tilpi := \pi\otimes\pi^\op$ of $\A_F\odot \A_F^\op$ on $\mH_R\oplus\mH_L$ is then given by 
$$
\tilpi((\lambda,q)\otimes(\lambda',q')^\op) = \lambda'q_\lambda \oplus \lambda'q .
$$
We define the \emph{mass matrix} on the basis $\{\nu_R,e_R,\nu_L,e_L\}$ as
$$
\D_F := \matfour{0&0&-im_\nu&0}{0&0&0&-im_e}{im_\nu&0&0&0}{0&im_e&0&0} .
$$
We then consider the even finite space $F\Sub{EW} := \left( \A_F , \mH_F , \D_F , \mJ_F = 1 \right)$.

\begin{prop}
\label{prop:gauge_gp_EW}
The gauge group of $F\Sub{EW}$ equals 
$$
\G(F\Sub{EW}) = \big(U(1) \times SU(2)\big) / \Z_2 .
$$ 
\end{prop}
\begin{proof}
We have $\mU(\A_F) = U(1)\times SU(2)$. The kernel of $\rho = \tilpi\circ\Delta\colon\mU(\A_F)\to\B(\mH_F)$ equals $\Ker\rho = \{ (\pm1,\pm1)\in\mU(\A_F) \} \simeq \Z_2$. The quotient $\G(F\Sub{ED}) = \mU(\A_F) / \Ker\rho$ is thus given by $\big(U(1) \times SU(2)\big) / \Z_2$. 
\end{proof}

Let $(M,g)$ again be an even-dimensional pseudo-Riemannian spin manifold as in Proposition \ref{prop:mfd_trip}, for which $t$ is odd. 
The representations $\pi$ and $\pi^\op$ obviously extend to representations of $C_c^\infty(M,\A_F)$ and $C_c^\infty(M,\A_F)^\op$ on $\mH_F\hot L^2(\bS)$, and it is easy to see that these representations satisfy the order-one condition \eqref{eq:order-one}. 
We consider the almost-commutative manifold
$$
F\Sub{EW} \times M := \left( C_c^\infty(M,\A_F\odot \A_F^\op) , \mH_F\hot L^2(\bS) , 1\hot i^t\sD + i\D_F\hot1 , 1\hot\mJ_M \right) .
$$

\begin{prop}
\label{prop:EW_inn_fluc}
The fluctuation of $\D := 1\hot i^t\sD + i\D_F\hot1$ by a perturbation $A\in\Pert(C_c^\infty(M,\A_F))$ is of the form
\begin{align*}
\D_A = \D + \eta_\D(A) = 1\hot i^t\sD + A_\mu \hot i^t\gamma^\mu + (i\D_F+\phi) \hot 1 ,
\end{align*}
where the \emph{gauge field} $A_\mu$ and the \emph{Higgs field} $\phi$ are given by
\begin{align*}
A_\mu &= \matthree{0&0&}{0&-2\Lambda_\mu&}{&&Q_\mu-\Lambda_\mu} , & 
\phi &= \matfour{0&0&m_\nu\bar{\phi_1}&m_\nu\bar{\phi_2}}{0&0&-m_e\phi_2&m_e\phi_1}{-m_\nu\phi_1&m_e\bar\phi_2&0&0}{-m_\nu\phi_2&-m_e\bar\phi_1&0&0} ,
\end{align*}
for the gauge fields $(\Lambda_\mu,Q_\mu) \in C_c^\infty\big(M,i\R\oplus\su(2)\big)$ and the Higgs field $(\phi_1,\phi_2) \in C_c^\infty(M,\C^2)$.
\end{prop}
\begin{proof}
Write $A = \sum_j a_j\otimes b_j^\op = \sum_j (\lambda_j,q_j)\otimes(\lambda'_j,q'_j)^\op \in \Pert(C_c^\infty(M,\A_F))$. 
Since the order-one condition is satisfied,
the fluctuation is of the form
$
\eta_\D(A) = a_j [\D,b_j] + a_j^{*\op} [\D,b_j^{*\op}] ,
$
where $\D := 1\hot i^t\sD + i\D_F\hot1$.  
The below calculations are similar to those in \cite[\S5.2.2]{vdDvS12}, and therefore we shall be brief. 
From the commutators with $i\D_F$, we obtain the \emph{Higgs field} 
\begin{align*}
\phi := \sum_j a_j [i\D_F,b_j] = \matfour{0&0&m_\nu\phi_1'&m_\nu\phi_2'}{0&0&-m_e\bar\phi_2'&m_e\bar\phi_1'}{-m_\nu\phi_1&m_e\bar\phi_2&0&0}{-m_\nu\phi_2&-m_e\bar\phi_1&0&0} ,
\end{align*}
where we define
\begin{align*}
\phi_1 &= \sum_j \alpha_j(\lambda_j'-\alpha_j')+\beta_j\bar\beta_j' , & 
\phi_2 &= \sum_j \bar\alpha_j\bar\beta_j'-\bar\beta_j(\lambda_j'-\alpha_j') , \\
\phi_1' &= \sum_j \lambda_j(\alpha_j'-\lambda_j') , & 
\phi_2' &= \sum_j \lambda_j\beta_j' .
\end{align*}
We also observe that $\D_F$ commutes with $\pi^\op$, so $a_j^{*\op} [i\D_F,b_j^{*\op}] = 0$. We note that reality of the perturbation $A$ ensures Krein-self-adjointness of $\phi$, and therefore self-adjointness of $i\phi$ (since $\phi$ anti-commutes with $\mJ$). Hence we must have $\phi_1' = \bar{\phi_1}$ and $\phi_2' = \bar{\phi_2}$. Furthermore, we can write
\begin{align*}
\sum_j a_j [i^t\sD,b_j] &= \matthree{\Lambda_\mu&0&}{0&-\Lambda_\mu&}{&&Q_\mu} \hot i^t\gamma^\mu , &
\sum_j a_j^{*\op} [i^t\sD,b_j^{*\op}] &= -\Lambda_\mu \hot i^t\gamma^\mu ,
\end{align*}
for $\Lambda_\mu := \sum_j \lambda_j\partial_\mu\lambda_j' \in C_c^\infty(M,i\R)$ and $Q_\mu := \sum_j q_j\partial_\mu q_j' \in C_c^\infty(M,\su(2))$. Thus the fluctuation of $\D = 1\hot i^t\sD + i\D_F\hot1$ by $A = a_j\otimes b_j^\op \in \Pert(\A)$ is of the form 
\begin{align*}
\eta_\D(A) = A_\mu \hot i^t\gamma^\mu + \phi \hot 1 ,
\end{align*}
where the \emph{gauge field} $A_\mu$ is given by
\begin{align*}
A_\mu &:= \matthree{0&0&}{0&-2\Lambda_\mu&}{&&Q_\mu-\Lambda_\mu} . \qedhere
\end{align*}
\end{proof}

The almost-commutative manifold $F\Sub{EW}\times M$ is of Lorentz-type, and hence we can apply Definition \ref{defn:Krein_action} for the Krein action. An arbitrary vector $\xi\in\mH^0 = \mH_L\hotimes L^2(\bS)^0 \oplus \mH_R\hotimes L^2(\bS)^1$ can be written uniquely as 
\begin{align}
\label{eq:right_fermion_EW}
\xi = \nu_R\hot\psi^\nu_R + e_R\hot\psi^e_R + \nu_L\hot\psi^\nu_L + e_L\hot\psi^e_L ,
\end{align}
for Weyl spinors $\psi^\nu_L,\psi^e_L\in L^2(\bS)^0$ and $\psi^\nu_R,\psi^e_R\in L^2(\bS)^1$. We observe that this vector $\xi\in\mH^0$ is completely determined by two Dirac spinors $\psi^\nu := \psi^\nu_L + \psi^\nu_R$ (describing the neutrino) and $\psi^e := \psi^e_L + \psi^e_R$ (describing the electron). 
We combine these spinors into the doublets of Weyl spinors 
\begin{align*}
\Psi_L &:= \vectwo{\psi^\nu_L}{\psi^e_L} \in L^2(\bS)^0\otimes\C^2 , & 
\Psi_R &:= \vectwo{\psi^\nu_R}{\psi^e_R} \in L^2(\bS)^1\otimes\C^2 , 
\end{align*}
and the corresponding doublet of Dirac spinors $\Psi := \Psi_L + \Psi_R \in L^2(\bS)\otimes\C^2$. 

\begin{prop}
\label{prop:alt_fermion_act_EW}
The Krein action for $F\Sub{EW}\times M$ is given by
\begin{multline*}
S\Sub{EW}[\Psi,A] 
= \la\Psi \mvert i^t\sD\Psi\ra + \la\psi^e_R \mvert {-2i^t\gamma^\mu\Lambda_\mu \psi^e_R} \ra + \la\Psi_L \mvert i^t\gamma^\mu(Q_\mu-\Lambda_\mu)\Psi_L\ra + \la\Psi_R \mvert \Phi \Psi_L \ra + \la\Psi_L \mvert \Phi^* \Psi_R \ra ,
\end{multline*}
where the gauge fields $\Lambda_\mu$ and $Q_\mu$ and the Higgs field $(\phi_1,\phi_2)$ are given in Proposition \ref{prop:EW_inn_fluc}, and the Higgs field $(\phi_1,\phi_2)$ acts via 
\begin{align*}
\Phi &:= \mattwo{-m_\nu\bar{(\phi_1+1)}}{-m_\nu\bar{\phi_2}}{m_e\phi_2}{-m_e(\phi_1+1)} .
\end{align*}
\end{prop}
\begin{proof}
We need to calculate the inner product $\la\mJ\xi,\D_A\xi\ra_\mJ$, where $\xi$ is given as in Eq.\ \ref{eq:right_fermion_EW}. First, for $\mJ = 1\hot\mJ_M$ we find
$$
\mJ\xi = - \nu_R\hot\mJ_M\psi^\nu_R - e_r\hot\mJ_M\psi^e_R + \nu_L\hot\mJ_M\psi^\nu_L + e_L\hot\mJ_M\psi^e_L .
$$
For the fluctuated Dirac operator $\D_A$ of Proposition \ref{prop:EW_inn_fluc} we find
\begin{align*}
\D_A\xi &= - \nu_R\hot i^t\sD\psi^\nu_R - e_R\hot i^t\sD\psi^e_R + \nu_L\hot i^t\sD\psi^\nu_L + e_L\hot i^t\sD\psi^e_L \\
&\quad + 2\Lambda_\mu e_R\hot i^t\gamma^\mu\psi^e_R + i^t\gamma^\mu(Q_\mu-\Lambda_\mu) (\nu_L\hot\psi^\nu_L,e_L\hot\psi^e_L)^t \\
&\quad - \nu_L\hot\big(m_\nu(\phi_1+1)\psi^\nu_R - m_e\bar{\phi_2}\psi^e_R\big) - e_L\hot\big(m_\nu\phi_2\psi^\nu_R + m_e\bar{(\phi_1+1)}\psi^e_R\big) \\
&\quad+ \nu_R\hot\big(m_\nu\bar{(\phi_1+1)}\psi^\nu_L + m_\nu\bar{\phi_2}\psi^e_L\big) - e_R\hot\big(m_e\phi_2\psi^\nu_L - m_e(\phi_1+1)\psi^e_L\big) .
\end{align*}
Taking the inner product of $\mJ\xi$ with $\D_A\xi$, and using the notation $\Psi_L$, $\Psi_R$, and $\Phi$, we obtain
\begin{align*}
\la\mJ\xi \mvert \D_A\xi\ra_{\mJ} 
&= \big\la\mJ_M\psi^\nu_R \bigm| i^t\sD\psi^\nu_R\big\ra_{\mJ_M} + \big\la\mJ_M\psi^e_R \bigm| i^t\sD\psi^e_R\big\ra_{\mJ_M} + \big\la\mJ_M\psi^\nu_L \bigm| i^t\sD\psi^\nu_L\big\ra_{\mJ_M} + \big\la\mJ_M\psi^e_L \bigm| i^t\sD\psi^e_L\big\ra_{\mJ_M} \\
&\quad- \big\la\mJ_M \psi^e_R \bigm| 2i^t\gamma^\mu\Lambda_\mu \psi^e_R \big\ra_{\mJ_M} + \big\la\mJ_M \Psi_L \bigm| i^t\gamma^\mu(Q_\mu-\Lambda_\mu)\Psi_L\big\ra_{\mJ_M} \\
&\quad + \big\la\mJ_M\Psi_R \bigm| \Phi \Psi_L \big\ra_{\mJ_M} + \big\la\mJ_M\Psi_L \bigm| \Phi^* \Psi_R \big\ra_{\mJ_M} ,
\end{align*}
The desired expression for $S\Sub{EW}[\Psi,A] = \la\xi \mvert \D_A\xi\ra = \la\mJ\xi \mvert \D_A\xi\ra_{\mJ}$ then follows by using the orthogonality of $L^2(\bS)^0$ and $L^2(\bS)^1$ and the symmetry of $\la\cdot|\cdot\ra$. 
\end{proof}

We observe that the Lagrangian calculated above is precisely (the fermionic part of) the usual Lagrangian for the lepton sector of the Glashow-Weinberg-Salam theory of electroweak interactions, including right-handed neutrinos (but without Majorana masses). For instance, the term $\la\Psi_L | \Phi^* \Psi_R \ra$ can be rewritten in the form
$$
- m_\nu \left\la \Psi_L \Biggm| \vectwo{\phi_1+1}{\phi_2} \psi^\nu_R \right\ra - m_e \left\la \Psi_L \Biggm| \vectwo{-\bar{\phi_2}}{\bar{\phi_1+1}} \psi^e_R \right\ra ,
$$
which is of the same form as \cite[Eq.(20.101)]{Peskin-Schroeder95} (though there it is given for quarks instead of leptons). If we substitute the vacuum expectation value for the Higgs field, setting $\phi_1+1=v/\sqrt{2}$ and $\phi_2=0$, we obtain
$$
-\frac1{\sqrt{2}} v \Bigl( m_\nu \left\la \psi^\nu_L \mvert \psi^\nu_R \right\ra + m_e \left\la \psi^e_L \mvert \psi^e_R \right\ra \Bigr) , 
$$
which are indeed the standard mass terms for the neutrino and the electron. 

\subsection{Majorana masses}
\label{sec:Majorana}

Let us briefly discuss how one can add Majorana masses into the model as well. For this purpose, we double up the Hilbert space, and introduce a \emph{real structure} (for the definition of a real structure, see e.g. \cite[Definition 3]{Con95b} or \cite[Definition 1.124]{CM07}).  
Given the Hilbert space $\mH_F$ with basis $\{\nu_R,e_R,\nu_L,e_L\}$, we create an identical copy $\mH_{\bar F}$ on which we denote the basis as $\{\bar{\nu_R},\bar{e_R},\bar{\nu_L},\bar{e_L}\}$, and we interpret this new copy as describing the anti-particles. We then consider the new Hilbert space $\hat\mH_F := \mH_F\oplus\mH_{\bar F}$ along with a mass matrix $\hat\D_F$, a `fundamental symmetry' $\hat\mJ_F$, a grading $\hat\Gamma_F$, and the \emph{real structure} $\hat J_F$ given by
\begin{align*}
\hat\D_F &:= \mattwo{\D_F}{-\D_M^*}{\D_M}{\bar{\D_F}} , & \hat\mJ_F &:= \mattwo{1}{0}{0}{-1} , &
\hat\Gamma_F &:= \mattwo{\Gamma_F}{0}{0}{-\Gamma_F} , & \hat J_F &:= \mattwo{0}{\cc}{\cc}{0} ,
\end{align*}
where $\cc$ denotes complex conjugation (with respect to the standard basis), and $\bar{\D_F}$ is the complex conjugate of the mass matrix $\D_F$. The map $\D_M\colon\mH_F\to\mH_{\bar F}$ is given as $\D_M \nu_R := i m_R \bar{\nu_R}$, 
where $m_R\in\R$ is the \emph{Majorana mass} of the right-handed neutrino, and $\D_M e_R = \D_M \nu_L = \D_M e_L = 0$. We point out that $\hat\D_F$ is Krein-self-adjoint, and that $\hat J_F$ anti-commutes with both $\hat\mJ_F$ and $\hat\Gamma_F$. The mass matrix $\hat\D_F$ does not commute (or anti-commute) with the real structure $\hat J_F$; 
instead we have the relation $\hat\D_F\hat J_F = \hat J_F\hat\D_F^*$ (where we used the symmetry of $\D_M$). Thus the commutator picks out the skew-adjoint part of $\hat\D_F$, which is just the part containing the Majorana mass. 
To be precise:
$$
[\hat D_F,\hat J_F] = \mattwo{-2\bar{\D_M}}{0}{0}{2\D_M} \circ \cc 
$$

Recalling the representations $\pi\colon \A_F\to\B(\mH_F)$ and $\pi^\op\colon \A_F^\op\to\B(\mH_F)$, we obtain representations $\hat\pi\colon \A_F\to\B(\mH_F\oplus\mH_{\bar F})$ and $\hat\pi^\op\colon \A_F^\op\to\B(\mH_F\oplus\mH_{\bar F})$ given by 
\begin{align*}
\hat\pi(a) &:= \pi(a) \oplus \pi^\op(a^t) , & \hat\pi^\op(a) &:= \hat J_F \hat\pi(a^*) \hat J_F ,
\end{align*}
where $a^t$ denotes the matrix transpose of $a$. With these definitions, we obtain a new finite space $\hat F\Sub{EW} := (\A_F\odot \A_F^\op, \hat\mH_F, \hat\D_F, \hat\mJ_F)$ with grading operator $\hat\Gamma_F$ and with a real structure $\hat J_F$. 

Now consider a $4$-dimensional Lorentzian spin manifold $M$. We also equip the Krein spectral triple over $M$ (given in Proposition \ref{prop:mfd_trip}) with a real structure, given by the \emph{charge conjugation operator} $J_M$ on the spinor bundle. 
For even-dimensional pseudo-Riemannian manifolds, the (anti)commutation relations of the charge conjugation operator are given in \cite[Proposition 3]{Bau94}. In the $4$-dimensional Lorentzian case, the charge conjugation operator commutes with the Clifford representation and with the Dirac operator $\sD$, anti-commutes with the grading operator $\Gamma_M$, and satisfies $J_M^2=-1$ . Since $J_M$ commutes with the Clifford representation, it commutes with the fundamental symmetry $\mJ_M = \gamma(e_0)$. 

Next, we can consider the almost-commutative manifold $\hat F\Sub{EW}\times M$, which we equip with the real structure $J := \hat J_F\hot J_M$. We calculate
\begin{align*}
J^2 &= (\hat J_F\hot J_M)(\hat J_F\hot J_M) = - \hat J_F^2 \hot J_M^2 = -1\hot(-1) = 1 , \\
J\Gamma &= (\hat J_F\hot J_M)(\hat\Gamma_F\hot\Gamma_M) = \hat J_F\hat\Gamma_F \hot J_M\Gamma_M = (-\hat\Gamma_F\hat J_F)\hot(-\Gamma_M J_M) = (\hat\Gamma_F\hot\Gamma_M)(\hat J_F\hot J_M) = \Gamma J , \\
J \mJ &= (\hat J_F\hot J_M)(\hat\mJ_F\hot\mJ_M) = \hat J_F\hat\mJ_F \hot J_M\mJ_M = (-\hat\mJ_F\hat J_F)\hot(\mJ_M J_M) = (\hat\mJ_F\hot\mJ_M)(\hat J_F\hot J_M) = \mJ J .
\end{align*}
We note that to get $[J,\mJ]=0$ (which we need below) we used that $\hat J_F \hat\mJ_F = -\hat\mJ_F \hat J_F$, and hence it is essential that $\mJ_F$ is non-trivial. 

Since we have doubled the finite-dimensional Hilbert space, we have introduced too many degrees of freedom. 
To correct this, we follow the approach of \cite{Bar07} and consider vectors $\eta\in\mH^0$ which (in addition to $\Gamma\eta=\eta$) also satisfy $J\eta=\eta$. 
Since $J^2=1$ and $J\Gamma=\Gamma J$, this assumption makes sense, and it means we can write $\eta = \xi + J\xi$, where $\xi$ is an element of $(\mH_F\hot L^2(\bS))^0$ as given in Eq.\ \ref{eq:right_fermion_EW}. 
The fermionic action is then of the form
$$
\la\mJ\eta \mvert \D_A\eta\ra_\mJ = \la\mJ\xi \mvert \D_A\xi\ra_\mJ + \la\mJ J\xi \mvert \D_A \xi\ra_\mJ + \la\mJ \xi \mvert \D_A J\xi\ra_\mJ + \la\mJ J\xi \mvert \D_A J\xi\ra_\mJ .
$$
The last term can be rewritten as
$$
\la\mJ J\xi \mvert \D_A J\xi\ra_\mJ = \la \mJ J\xi \mvert J\D_A\xi\ra_\mJ + \la J\mJ\xi \mvert [\D_A,J]\xi\ra_\mJ ,
$$
The commutator $[\D_A,J]$ equals $i[\hat\D_F,\hat J_F]\hot J_M$. For $\xi\in(\mH_F\hot L^2(\bS))^0$ we have $[\D_A,J]\xi\in\mH_F\hot L^2(\bS)$ while $J\mJ\xi\in\mH_{\bar F}\hot L^2(\bS)$. The subspaces $\mH_F\hot L^2(\bS)$ and $\mH_{\bar F}\hot L^2(\bS)$ are orthogonal, so the term $\la J\mJ\xi \mvert [\D_A,J]\xi\ra_\mJ$ vanishes. Since $J$ commutes with $\mJ$, we find that $\la\mJ J\xi \mvert \D_A J\xi\ra_\mJ = \la\mJ\xi \mvert \D_A\xi\ra_\mJ = S\Sub{EW}[\Psi,A]$. 

Hence the only new contributions to the fermionic action come from the terms $\la\mJ J\xi \mvert \D_A \xi\ra_\mJ$ and $\la\mJ \xi \mvert \D_A J\xi\ra_\mJ$. Since the subspaces $\mH_F\hot L^2(\bS)$ and $\mH_{\bar F}\hot L^2(\bS)$ are orthogonal, we only need to consider the part of $\D_A$ which mixes particles and anti-particles, which is precisely just the Majorana mass matrix $\D_M$. For the vector $\xi^\nu_R := \nu_R\hot\psi^\nu_R$ representing the right-handed neutrino, we calculate
\begin{align*}
J \xi^\nu_R &= - \hat J_F\nu_R \hot J_M\psi^\nu_R = - \bar{\nu_R} \hot J_M\psi^\nu_R , \\
\mJ J \xi^\nu_R &= - \mJ (\bar{\nu_R} \hot J_M\psi^\nu_R) = - \hat\mJ_F\bar{\nu_R} \hot \mJ_MJ_M\psi^\nu_R = \bar{\nu_R} \hot \mJ_MJ_M\psi^\nu_R , \\
(i\D_M\hot1)\xi^\nu_R &= -m_R \bar{\nu_R}\hot\psi^\nu_R , \\
(-i\D_M^*\hot1) J \xi^\nu_R &= m_R \nu_R\hot J_M\psi^\nu_R .
\end{align*}
This gives
\begin{align*}
\la\mJ\xi | \D_A J\xi\ra_\mJ &= \big\la {-\nu_R\hot\mJ_M\psi^\nu_R} | m_R\nu_R\hot J_M\psi^\nu_R \big\ra_\mJ = -m_R \big\la \mJ_M\psi^\nu_R | J_M\psi^\nu_R \big\ra_{\mJ_M} , \\*
\la \mJ J\xi | \D_A \xi\ra_\mJ &= \big\la \bar{\nu_R} \hot \mJ_MJ_M\psi^\nu_R | {-m_R \bar{\nu_R}\hot\psi^\nu_R} \big\ra_\mJ = -m_R \big\la \mJ_M J_M \psi^\nu_R | \psi^\nu_R \big\ra_{\mJ_M} .
\end{align*}
Summarising, we can extend the electro-weak theory to include Majorana masses for right-handed neutrinos, and we obtain the new action $S\Sub{EW\!+\!M}$ given by
$$
S\Sub{EW\!+\!M}[\Psi,A] = 2 S\Sub{EW}[\Psi,A] - m_R \la \psi^\nu_R | J_M\psi^\nu_R \ra - m_R \la J_M \psi^\nu_R | \psi^\nu_R \ra ,
$$
where $S\Sub{EW}$ is given in Proposition \ref{prop:alt_fermion_act_EW}. 

\begin{remark}
In this section we have introduced the notion of real structure on a Krein spectral triple in a rather ad hoc manner. One might wonder how to describe the general theory of real Krein spectral triples. Do we for instance obtain an alternative notion for the KO-dimension of such triples? While the KO-dimension of real spectral triples arises from the eightfold Bott periodicity in real K-homology, there is no Krein version of K-homology, and hence the answer to this question is not straightforward. In any case, one would expect that the signs determining the (anti)commutation relations of the real structure now not only depend on the dimension, but also on the signature. For the even-dimensional case, these signs are given in \cite{Bau94}. 
\end{remark}

\section{The Standard Model}
\label{sec:krein_SM}

Given the description of the electro-weak theory of the previous section, it is fairly straightforward to extend this theory to the full Standard Model as described in \cite{Con06,CCM07}. 
This extension is basically obtained by including a summand $M_3(\C)$ in the algebra $\A_F$ to describe the strong interactions, and by enlarging the Hilbert space $\mH_F$ to incorporate the quarks. Moreover, the Hilbert space is then enlarged three-fold to include three generations of all elementary particles. Since most of the details are similar to the electro-weak theory, and since there is already plenty of literature available on the noncommutative description of the Standard Model (see e.g.\ \cite{Con96,Con06,CCM07,CM07,JKSS07,CC10,vdDvS12}), we shall be rather brief in this section.

Thus, we take the algebra $\A_F = \C\oplus\qH\oplus M_3(\C)$, which is represented on the finite-dimensional Hilbert space $\mH_F := (\mH_R\oplus\mH_L) \otimes \C^3$. The factor $\C^3$ describes the fact that there are three generations of elementary particles. The right-handed particles $\mH_R$ and the left-handed particles $\mH_L$ are both given by $\C^2 \oplus (\C^2\otimes\C^3)$. Here the first summand $\C^2$ describes the two leptons $\nu$ and $e$, and the second summand $\C^2\otimes\C^3$ describes the quarks $u^c$ and $d^c$ (which occur in three colours $c=r,g,b$). 

We will consider the commuting representations $\pi\colon \A_F\to\B((\mH_R\oplus\mH_L) \otimes \C^3)$ and $\pi^\op\colon \A_F^\op\to\B((\mH_R\oplus\mH_L) \otimes \C^3)$ given by
\begin{align*}
\pi(\lambda,q,b) &:= \big( (q_\lambda \oplus (q_\lambda\otimes1)) \oplus (q\oplus (q\otimes1)) \big) \otimes 1 , \\
\pi^\op(\lambda,q,b) &:= \big( (\lambda \oplus (1\otimes b^t)) \oplus (\lambda \oplus (1\otimes b^t)) \big) \otimes 1 ,
\end{align*}
where $b^t$ denotes the matrix transpose of $b$. 
The representation $\tilpi:=\pi\otimes\pi^\op\colon \A_F\odot \A_F^\op \to \B(\mH_F)$ is then given by
$$
\tilpi\big((\lambda,q,b)\otimes(\lambda',q',b')^\op\big) = \big( (\lambda'q_\lambda \oplus (q_\lambda\otimes b'^t)) \oplus (\lambda'q\oplus (q\otimes b'^t)) \big) \otimes 1 .
$$
We consider the even finite space $F\Sub{SM} := \left( \A_F , \mH_F , \D_F , \mJ_F = 1 \right)$, where the mass matrix is given by
\begin{align*}
\D_F &:= \matfour{0&0&-iY_\nu&0}{0&0&0&-iY_e}{iY_\nu&0&0&0}{0&iY_e&0&0} \oplus \matfour{0&0&-iY_u&0}{0&0&0&-iY_d}{iY_u&0&0&0}{0&iY_d&0&0} \otimes1 .
\end{align*}
Here each $Y_\bullet$ is a hermitian $3\times3$-matrix corresponding to the three generations of each type of particle. 

Similarly to Proposition \ref{prop:gauge_gp_EW}, the gauge group of the finite space $F\Sub{SM}$ is given by $\G(F\Sub{SM}) = \big(U(1)\times SU(2)\times U(3)\big) / \Z_2$. This gauge group does not match the gauge group of the Standard Model (even modulo finite groups), since we have a factor $U(3)$ instead of $SU(3)$. 
As in \cite[\S2.5]{CCM07} (see also \cite[\S6.2.1]{vdDvS12}), we will therefore impose the \emph{unimodularity condition} $\det|_{\mH_F}\big(\rho(u)\big) = 1$, which yields the subgroup
\begin{align*}
S\G(F\Sub{SM}) = \Big\{ \rho(u) \in \G(F\Sub{SM}) : u = (\lambda,q,b)\in\mU(\A_F) ,\; \big(\lambda\det b\big)^{12} = 1 \Big\} .
\end{align*}
The effect of the unimodularity condition is that the determinant of $b\in U(3)$ is identified (modulo the finite group $\mu_{12}$ of $12$th-roots of unity) to $\bar\lambda\in U(1)$. In other words, imposing the unimodularity condition provides us, modulo some finite abelian group, with the gauge group $U(1)\times SU(2) \times SU(3)$ of the Standard Model. 

The calculations for the inner fluctuations and the fermionic action of the Standard Model are similar to the case of the electro-weak theory (although somewhat more cumbersome). Below we will simply give the results. 

\begin{prop}
\label{prop:SM_inn_fluc}
The fluctuation of $\D$ by $A\in\Pert(C_c^\infty(M,\A_F))$ is of the form
\begin{align*}
\D_A = \D + \eta_\D(A) = 1\hot i^t\sD + A_\mu \hot i^t\gamma^\mu + (i\D_F+\phi) \hot 1 ,
\end{align*}
where the \emph{gauge field} $A_\mu$ and the \emph{Higgs field} $\phi$ are given by
\begin{align*}
A_\mu &= \matthree{0&0&}{0&-2\Lambda_\mu&}{&&Q_\mu-\Lambda_\mu} \oplus \left(\matthree{\frac43\Lambda_\mu&0&}{0&-\frac23\Lambda_\mu&}{&&(Q_\mu+\frac13\Lambda_\mu)\otimes1} + 1\otimes V_\mu\right) \\
\phi &= \matfour{0&0&Y_\nu\bar{\phi_1}&Y_\nu\bar{\phi_2}}{0&0&-Y_e\phi_2&Y_e\phi_1}{-Y_\nu\phi_1&Y_e\bar\phi_2&0&0}{-Y_\nu\phi_2&-Y_e\bar\phi_1&0&0} \oplus \matfour{0&0&Y_u\bar{\phi_1}&Y_u\bar{\phi_2}}{0&0&-Y_d\phi_2&Y_d\phi_1}{-Y_u\phi_1&Y_d\bar\phi_2&0&0}{-Y_u\phi_2&-Y_d\bar\phi_1&0&0} \otimes1,
\end{align*}
for the gauge fields $(\Lambda_\mu,Q_\mu,V_\mu) \in C_c^\infty\big(M,i\R\oplus\su(2)\oplus\su(3)\big)$ and the Higgs field $(\phi_1,\phi_2) \in C_c^\infty(M,\C^2)$.
\end{prop}

Similarly to Eq.\ \ref{eq:right_fermion_EW}, an arbitrary vector $\xi\in\mH^0 = \mH_L\otimes L^2(\bS)^0 \oplus \mH_R\otimes L^2(\bS)^1$ is uniquely determined by Dirac spinors $\psi^\nu$ (describing the three neutrinos), $\psi^e$ (describing the electron, muon, and tau-particle), $\psi^u$ (describing the up, charm, and top quarks in three colours) and $\psi^d$ (describing the down, strange, and bottom quarks in three colours), where we have omitted the generational index from our notation. We group these spinors together into the multiplets $\Psi^l \in L^2(\bS)\otimes\C^2\otimes\C^3$ (describing the leptons) and $\Psi^q \in L^2(\bS)\otimes\C^2\otimes\C^3\otimes\C^3$ (describing the quarks). 

\begin{prop}
\label{prop:alt_fermion_act_SM}
The Krein action for $F\Sub{SM}\times M$ is given by
\begin{align*}
S\Sub{EW}[\Psi,A] 
&= \la\Psi^l | i^t\sD\Psi^l\ra + \la\Psi^q | i^t\sD\Psi^q\ra \\
&\quad+ \la\psi^e_R | {-}2i^t\gamma^\mu\Lambda_\mu \psi^e_R \ra + \la\psi^u_R | \tfrac43i^t\gamma^\mu\Lambda_\mu \psi^u_R \ra + \la\psi^d_R | {-}\tfrac23i^t\gamma^\mu\Lambda_\mu \psi^d_R \ra \\
&\quad+ \la\Psi^l_L | i^t\gamma^\mu(Q_\mu-\Lambda_\mu)\Psi^l_L\ra + \la\Psi^q_L | i^t\gamma^\mu(Q_\mu-\Lambda_\mu)\Psi^q_L\ra + \la\Psi^q | V_\mu \Psi^q\ra \\
&\quad+ \la\Psi^l_R | \Phi^l \Psi^l_L \ra + \la\Psi^l_L | (\Phi^l)^* \Psi^l_R \ra + \la\Psi^q_R | \Phi^q \Psi^q_L \ra + \la\Psi^q_L | (\Phi^q)^* \Psi^q_R \ra ,
\end{align*}
where the gauge fields $\Lambda_\mu$, $Q_\mu$, and $V_\mu$ and the Higgs field $(\phi_1,\phi_2)$ are given in Proposition \ref{prop:SM_inn_fluc}, and the Higgs field acts via
\begin{align*}
\Phi^l &:= \mattwo{-Y_\nu\bar{(\phi_1+1)}}{-Y_\nu\bar{\phi_2}}{Y_e\phi_2}{-Y_e(\phi_1+1)} , &
\Phi^q &:= \mattwo{-Y_u\bar{(\phi_1+1)}}{-Y_u\bar{\phi_2}}{Y_d\phi_2}{-Y_d(\phi_1+1)} .
\end{align*}
\end{prop}

We observe that the Lagrangian calculated above is precisely (the fermionic part of) the usual Lagrangian for the Standard Model, including right-handed neutrinos (but without Majorana masses). 
For the possible inclusion of Majorana masses for the right-handed neutrinos, the procedure is the same as in Section \ref{sec:Majorana}, and we shall not repeat it here.

% remove vertical spacing between entries of bibliography (and modify label width)
\let\OLDthebibliography\thebibliography
\renewcommand\thebibliography[1]{
  \OLDthebibliography{JKSS07}
  \setlength{\parskip}{0pt}
  \setlength{\itemsep}{0pt plus 0.3ex}
}

% \bibliographystyle{myamsalpha}
% \bibliography{short,bibliography}

\providecommand{\noopsort}[1]{}\providecommand{\vannoopsort}[1]{}
\providecommand{\bysame}{\leavevmode\hbox to3em{\hrulefill}\thinspace}
\providecommand{\MR}{\relax\ifhmode\unskip\space\fi MR }
% \MRhref is called by the amsart/book/proc definition of \MR.
\providecommand{\MRhref}[2]{%
  \href{http://www.ams.org/mathscinet-getitem?mr=#1}{#2}
}
\providecommand{\href}[2]{#2}

\end{document}